\documentclass{sig-alternate}
\pdfoutput=1

\usepackage{listings,amsmath}
\usepackage{multirow}
\usepackage{nicefrac}
\usepackage{balance}
\usepackage{graphicx}
\begin{document}

\newtheorem{invariant}{Invariant}
\newtheorem{theorem}{Theorem}

\clubpenalty=10000

\widowpenalty = 10000 

\title{Ad Serving Using a Compact Allocation Plan}

\numberofauthors{3}
\author{
\alignauthor
Peiji Chen\\
    \affaddr{Yahoo!Labs}\\
    \email{peiji@yahoo-inc.com}
\alignauthor
Wenjing Ma\\
    \affaddr{Yahoo!Labs}\\
    \email{wenjingm@yahoo-inc.com}
\alignauthor
Srinath Mandalapu\\
    \affaddr{Yahoo!Labs}\\
    \email{srinathm@yahoo-inc.com}
\and
\alignauthor
Chandrashekhar Nagarajan\\
    \affaddr{Yahoo!Labs}\\
    \email{cn54@yahoo-inc.com}
\alignauthor
Jayavel Shanmugasundaram\\
    \affaddr{Google}\thanks{Work was done while with Yahoo! Labs} \\ 
    \email{jayavel.shanmugasundaram@acm.org}
\alignauthor
Sergei Vassilvitskii\\
    \affaddr{Yahoo!Labs}\\
    \email{sergei@yahoo-inc.com}
\and
\alignauthor
Erik Vee\\
    \affaddr{Yahoo!Labs}\\
    \email{erikvee@yahoo-inc.com}
\alignauthor
Manfai Yu\\
    \affaddr{Yahoo!Labs}\\
    \email{manfai@yahoo-inc.com}
\alignauthor
Jason Zien \\
    \affaddr{Yahoo!Labs}\\
    \email{jasonyz@yahoo-inc.com}
}


\maketitle

\begin{abstract}
A large fraction of online display advertising is sold via guaranteed
contracts: a publisher guarantees to the advertiser a certain number
of user visits satisfying the targeting predicates of the contract.
The publisher is then tasked with solving the ad serving problem
---given a user visit, which of the thousands of matching contracts
should be displayed, so that by the expiration time every contract has
obtained the requisite number of user visits. The challenges of the
problem come from (1) the sheer size of the problem being solved, with
tens of thousands of contracts and billions of user visits, (2) the
unpredictability of user behavior, since these contracts are sold
months ahead of time, when only a forecast of user visits is available
and (3) the minute amount of resources available online, as an ad
server must respond with a matching contract in a fraction of a
second.

We present a solution to the guaranteed delivery ad serving problem using {\em compact allocation plans}.
These plans, computed offline,
can be efficiently queried by the ad server during an ad call; they
are small, using only $O(1)$ space for contract; and are stateless,
allowing for distributed serving without any central
coordination. We evaluate this approach on a
real set of user visits and guaranteed contracts and show that the
compact allocation plans are an effective way of solving the
guaranteed delivery ad serving problem.
\end{abstract}

\newcommand{\ssum}[2][]{\mbox{\ensuremath{\sum^{{#1}}_{{#2}}}\,}}
\newcommand{\neij}[1][j]{\Gamma({#1})}
\newcommand{\neii}[1][i]{\Gamma({#1})}
\newcommand{\rate}{\alpha}

\section{Introduction}
\label{sec:intro}
\newcommand{\ignore}[1]{}

A central problem facing online advertising systems is {\em ad
serving}, i.e., deciding how to rapidly match supply (user visits) and
demand (ads) in a way that meets an overall objective. Our focus in
this paper is on guaranteed delivery display advertising, a multi-billion
dollar industry, whereby advertisers can buy a fixed number of
targeted user visits (e.g., 100 million user visits by Males in
California visiting Sports web pages) for a future duration
(e.g. July-August 2011), and internet publishers {\em guarantee} these
visits months in advance. In such cases, when the user visits (supply)
actually occur, the publisher's ad server is faced with a split-second
decision of choosing an appropriate guaranteed contract (demand) in
such a way that in the aggregate, all the guarantees are met. In
addition, as part of the split-second decision, the ad server also
needs to ensure some secondary objectives such as uniformly serving
user visits to guaranteed contracts; e.g., if there is a guaranteed
contract that spans two months, the advertiser expects their ad to be
shown throughout the duration, rather than just being shown to user
visits that occur during the first few days.

The above scenario can be modeled as an assignment problem. We are
presented with a set of user visits $I$, a set of guaranteed
contracts $J$, and a set of feasible allocation edges $E \subseteq I
\times J$, which indicates which user visits are eligible to be
served by which guaranteed contracts. Each guaranteed contract $j
\in J$ has a demand $d_j$, and we are asked to allocate the user
visits to satisfy the demands. In addition, there is an objective
function over the feasible allocations, such as uniform delivery
over the duration of the contracts, and the goal is to find the
allocation that maximizes the objective function.

Solving the above problem even in an {\em offline} setting presents a
formidable set of challenges. First, the set $I$ is extremely large,
approaching tens of billions of user visits every day for large publishers; since we may
have to solve the problem for up to a year in advance in the case
of guaranteed contracts, this results in an effective set size of tens
to hundreds of trillions! (Note that the set of edges $E \subseteq I
\times J$ is even larger!) Second, we often do not know the set $I$
exactly, but only approximately; for instance, we do not know
precisely which users will visit a particular web page, but only have
a stochastic estimate of the set based on past history.

Of course, our actual situation is even more challenging because the
ad server needs to solve the above problem {\em online}, usually
within a hundred milliseconds per request, in order to meet the
latencies for end-user facing actions such as page views. Further, in
order to meet the throughput requirements, the ad server needs to
solve the online problem in parallel among hundreds to thousands of
distributed machines that do not communicate directly with each other
for each user visit,
while still ensuring that the online decisions attempt to maximize the overall objective.
Finally, due to the long-tailed nature of user
behavior, the ad server needs to explicitly deal with the fact that it
may encounter user visits that have not been seen before, but must
nevertheless ensure that the online decisions attempt to maximize the
overall objective. These are a daunting set of issues --- what do we
do?

One simple approach is to design an ad server that selects a contract based on how many impressions (i.e. user visits) have been
previously served to each eligible contract. For instance, if there are two competing contracts, one of which is severely
under-delivering and the other which is delivering well (or even slightly over-delivering), the ad server may decide to serve
the under-delivering contract. While this approach is conceptually simple and practical (in fact, it serves as the baseline in our
experiments), it has a number of subtle problems. First, the above approach requires a good definition of the notion of ``under-delivering.''
The most straightforward method is simply to assume that contracts should be served at a uniform rate --- a contract with demand of
30 million impressions and a duration of 30 days should deliver about 1 million impressions a day.
This has the advantage of making delivery smooth, when possible. But it ignores the fact that the actual traffic
varies greatly over time. Weekends have less traffic than weekdays, while the traffic at night is typically much smaller than the
traffic during the day. In addition, there may be an upcoming dearth of impressions available for contracts, which may require the
ad server to deviate from the uniform rate in order to meet the guarantee. This can occur, for instance, when some type of user visits
are promised exclusively for premium contracts at some time in the future (e.g., a sports company buys out all visits to Yahoo! Super
Bowl pages); this may require the ad server to deliver more than the uniform rate during the days before the event in order to meet
guarantees. Second, using the above approach, the ad servers need to maintain some notion of how much a contract has delivered in
(near) real-time so that they can compute under-delivery at any point in time. In cases where there are thousands of ad servers
distributed across the world, this problem becomes quite challenging, especially when user visits exhibit certain locality patterns
that cause non-uniform traffic to be sent to different ad servers.

Another approach~\cite{EC} that addresses many of the
above issues is to solve an allocation problem offline, and then send a rate-based compact allocation plan to the distributed ad
servers. The allocation plan is based on a forecast of user visits over the lifetime of the contracts, so it effectively addresses
the issue of what it means for a contract to under-deliver. (It is when it does not meet the goal at the end of its lifetime, even if
it is above or below the uniform serving rate currently). Further, for a certain class of objectives, the allocation plan is rate-based,
meaning that the ad server does not rely on any online counts of how much a contract has served in the immediate past, but is based on
a serving probability that does not require online state~\cite{EC}. However, this approach also has a set of
issues that need to be addressed. First, this approach is very sensitive to the accuracy of the forecast of user visits. For instance,
if the forecast for a certain set of user visits is twice a large, then the the serving rate will be half as large as it should be.
Second, this approach typically requires solving a non-linear convex optimization problem, which can have on the order of billions of
variables, which can be quite expensive in terms of time and resources.

\subsection{Contributions}
In order to address the above issues with ad serving, we propose the
following two-part solution. First, we introduce a novel,
lightweight algorithm for solving the offline allocation problem.
The algorithm, which we call High Water Mark (HWM), produces a
compact allocation plan, requiring O(1) numbers per contract.
Further, it is stateless and rate-based, which ensures that multiple
ad servers can execute the algorithm in parallel without any
bottlenecks. Although it does not have the strong theoretical
guarantees of ~\cite{EC}, it has the advantage of quick
optimization. (Finding an allocation plan takes only a few minutes,
even with a billion arcs in the demand-supply forecast graph.)
Second, we introduce a rapid feedback loop from the ad server to the
optimization system, which reacts to recent aggregate delivery
statistics in order to correct for any supply forecasting errors. In
fact, we can even provide theoretical bounds on how the frequency of
feedback enables the HWM algorithm to correct for supply forecasting
errors (these bounds are also borne out experimentally). Thus, the
feedback loop, coupled with the quick HWM optimization, enables the
ad serving system to quickly react to unforeseen effects, while
still effectively addressing supply and demand constraints that vary
over the duration of the contracts.

In addition, we briefly describe an implementation of a variant of
the algorithm of~\cite{EC} as a point of comparison, which we call
the {\em DUAL algorithm}.

We present a thorough evaluation of the HWM algorithm, comparing it
to a real world baseline which is currently responsible for serving
a large fraction of Yahoo!'s guaranteed contracts. Unlike the HWM
algorithm, which relies on forecasts of user supply to make
allocation decisions, the baseline algorithm is a well-tuned
feedback based approach. At a very high level, it increases the
serving rate for the underdelivering contracts, and, vice-versa,
decreases it for those serving too aggressively. We show that the
proposed algorithm greatly improves overall delivery of the
contracts, while maintain a smooth serving rate.

We further present experiments comparing HWM and DUAL.
We see that the two perform effectively the same--- the
under-delivery for HWM is on par with DUAL, and its smoothness is
actually better. Thus, HWM is a good choice in practice, due to its
lightweight implementation and quick optimization cycle time.

\subsection{Paper Outline}
We formally define the guaranteed delivery allocation problem in
Section \ref{sec:model} and present an overview and the system
architecture in Section \ref{sec:asp}. We state the proposed HWM
offline algorithm along with the online HWM allocation algorithms in
Section~\ref{sec:offlineAlgo}. In Section~\ref{sec:exp} we describe
our experiment and metrics.  In Section~\ref{sec:hwm and base}, we
thoroughly evaluate the proposed approach against a real world
production baseline. Then, in Section~\ref{sec:hwm and dual}, we
compare HWM and DUAL. We review the related work in Section
\ref{sec:related} and conclude with some open problems in
Section~\ref{sec:conclusion}.

\section{Model and Problem Statement}
\label{sec:model}

We begin by providing a detailed example of the guaranteed delivery display advertising problem.
Recall that we are presented with a set of
advertisers, each aiming to reach a specific number of user visits
(contract demand) that fall under the targeting parameters.

Consider the problem pictorially represented in Figure~\ref{fig:example}. There are three advertisers, targeting Males,
users from California, and users with Age in the 5th bucket. There are
6 different kinds of users, all of them satisfying the age
constraint. Some of the users are known to be from California,
Washington or Nevada; while for others their location is
unknown. Similarly, some the publisher knows are male, while for
others the gender is unknown. Each node on the supply side is
annotated with the total forecasted number of visits, for example we
expect to see 100,000 user visits from Californians with age in the
5th bucket and unknown gender. On the demand side each contract is
annotated with the total number of user visits guaranteed to it by the
publisher. An edge between a supply node $i$ and a demand node $j$
indicates that supply $i$ is {\em eligible} for $j$, that is it
satisfies all of the targeting criteria.

\begin{figure}[t]
\centering
\includegraphics[width=3in]{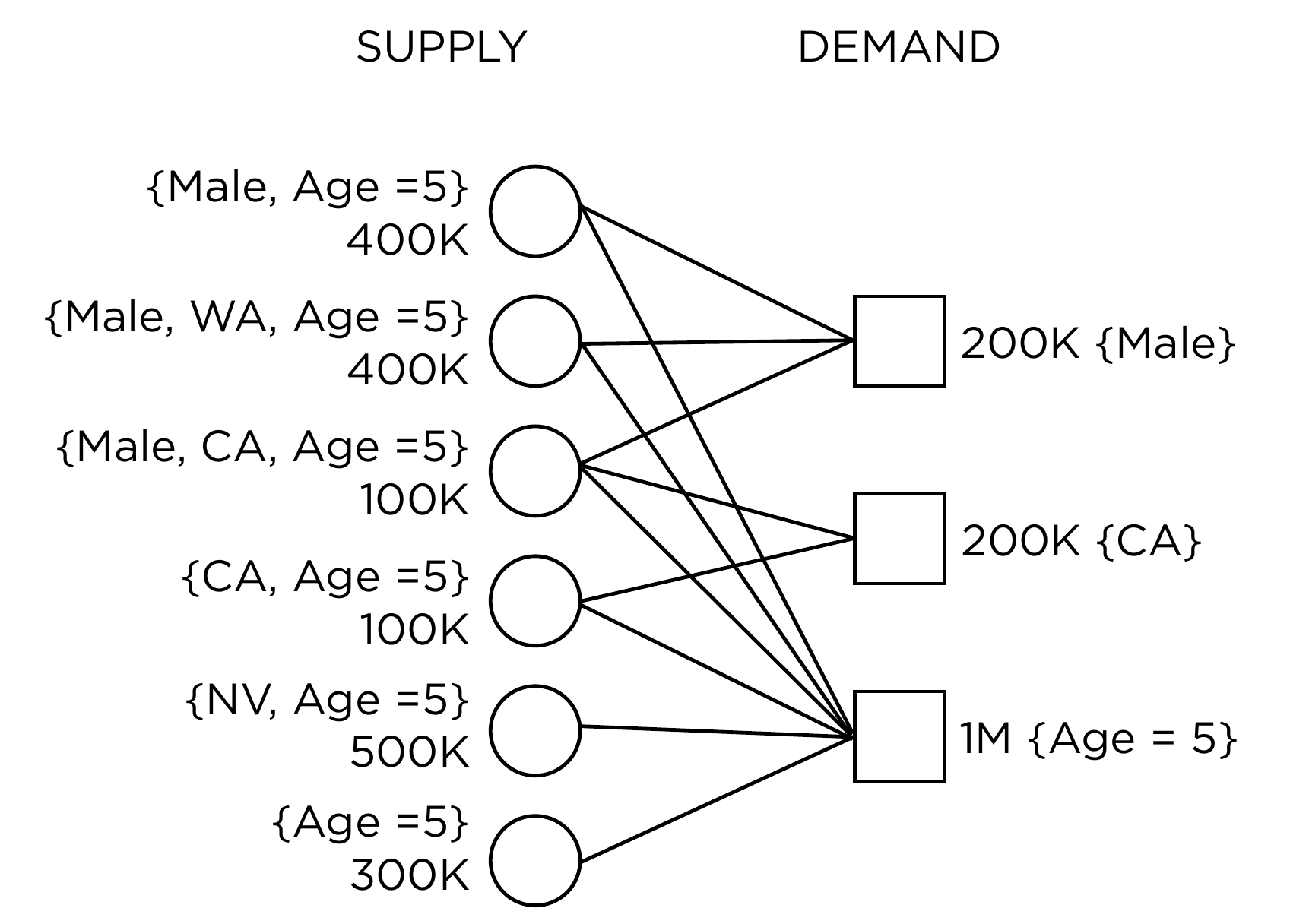}
\caption{An example allocation problem. Each user (left) is connected to the contracts that it is eligible for (right). }
\label{fig:example}
\end{figure}

The ad server must find an allocation of each supply node to the
demand. For example, to satisfy the second contract both the third
(Male, CA, Age = 5) and the fourth (CA, Age = 5) nodes must be fully
allocated to the second contract. When a Male user from
California with Age=5 arrives it is eligible for all three of the
contracts. However, it must be allocated to the second contract;
otherwise the second contract will not reach its desired demand and
the publisher will be hit with underdelivery penalties.

We note that while the above example is a small illustration, the
actual problem is much bigger in at least three dimensions. First, the
number of user visits is of the order of billions per {\em day}, and
guaranteed contracts can be sold several {\em months} in
advance. Second, the number of attributes per user visit is of the order of
hundreds, including demographic information such as age, gender, location,
explicit interest, inferred interests, recent browsing
behavior, and so on; page information such as the topic and ad-size; and
time information including date, user time and global time.
Finally, the number of guaranteed contracts is of
the order of tens to hundreds of thousands for large
publishers. Consequently, the scale of the graph --- both in terms of
supply and demand --- is quite large.  Also, in practice, the
advertiser targeting constraints are complex boolean
expressions on the attributes of the supply.

\subsection{Problem Statement}
\label{sec:prob statement}
 Let $I$ be the set of [forecasted] user
visits, and $J$ be the set of contracts. We model the allocation
problem as a bipartite graph $G = (I \cup J, E)$ where there is an
edge $(i,j) \in E$ if the user visit $i$ can be used to satisfy
contract $j$. The bipartite graph encodes the contention between
contracts, for example, a node $j \in J$ with few outgoing edges
means that it is a finely targeted contract, and there are few users
satisfying the targeting constraints. Conversely, a node $i \in I$
with a high degree means that it is a highly contended user, that is
there are many contracts that wish to advertise to this particular
user.

In addition to the targeting constraints, each contract $j \in J$
specifies its desired demand $d_j$. Likewise, a user may have
multiple identical visits, so each node $i \in I$ is labeled with
the total supply, $s_i$.

We wish to find a feasible allocation, that is a value on each edge
$x_{ij}$, representing the fraction of supply $i$ is allocated
to contract $j$. We say that $x$ is a feasible solution if it
satisfies:
\begin{align*}\itemsep=0in
\mbox{\ \ }
        & \forall_j \ \ \ \ssum{i \in \neij} x_{ij} s_i  \ge d_j & \mbox{{\em demand constraints}} \\
        & \forall_i, \ \ \ssum{j\in \neii} x_{ij} \leq 1 & \mbox{{\em supply constraints}}\\
        & \forall_{(i,j)\in E},\ \  x_{ij} \geq  0 & \mbox{{\em non-negativity constraints}},
\end{align*}
where following standard notation, $\neii$ are the neighbors of $i$:
$\neii = \{j : (i,j) \in E \}$, and similarly for $\neij$.

Although we do not explicitly require it, one important aspect of serving is
smoothness of delivery.  Advertisers do not want their entire demand to be served in a single
hour (nor do users!), even if this technical satisfies the demand.  Thus, any serving method
should roughly serve an equal number of impression each day for a given contract, subject
to natural variations in contention, sell-through rate, and overall traffic patterns.

There are, of course, potentially many additional secondary considerations.
For a longer description of the formal model and possible competing objectives of the offline
optimization system see \cite{Informs,CIKM,Ghosh,KDD}.

\section{Solution Overview and System Architecture}
\label{sec:asp}

Before laying out our solution, we first discuss several alternatives.

Perhaps the simplest ad serving solution is to randomly flip an
unbiased coin between all the matching contracts. This solution is
fast, space-efficient, has a large throughput, requires no
coordination, and generalizes to previously unseen input. Unfortunately,
what it gains in speed and flexibility, it loses in accuracy. Indeed,
because this approach makes no distinction between contracts during
contract selection, it will typically lead to severe under-delivery
and thus loss of revenue. For instance, in the example in
Figure~\ref{fig:example}, all Males in CA in Age Group 5 should be
allocated to the second contract even though it is eligible for all
three contracts, or the second contract will
under-deliver.

An alternative solution that tries to maximize accuracy is as
follows. Solve the allocation problem offline, and remember the
allocation values $x_{ij}$ on all the arcs of the bipartite graphs in
the ad server machines. Then, given a user visit, the ad server can
simply look up the allocation and assign it to the appropriate
contract. While this would certainly be accurate (it
implements the optimal allocation), it is very memory intensive given
the large number of nodes and edges that need to be stored. Further,
and perhaps more importantly, it is not generalizable. Specifically,
it is impractical to assume that we can predict all the possible user
visit types that could {\em possibly} occur (recall that it is a
combination of dynamic user characteristics such as interests, along
with their IP location, pages they visit, etc.) and solve an
optimization problem for all such visits. As an illustration, consider
again the example in Figure \ref{fig:example}, and a Female user from
California with Age=5. Since this type of user was {\em not}
considered in computing the offline allocation, there are no $x_{ij}$
values present. Therefore, while there are two eligible contracts, the
ad server has no information as to which contract should be shown. In
this example we would prefer to show the ``CA" contract, since it is
harder to satisfy.  However, this information is not easy to elicit
from the $x_{ij}$ values.

Intuitively, we want to get the best of both approaches, obtaining the
accuracy of the second approach, while retaining the simplicity and
scalability of the first approach. In Section~\ref{sec:offlineAlgo}, we present the
overview of a solution that achieves these goals.

\subsection{System Architecture}

Figure~\ref{fig:arch} shows the proposed system architecture. As
shown, there is an offline component called the Allocation Plan
Generator that takes in a forecast of user visits and booked
contracts, and solves the allocation problem. We solve the problem on a {\em sample} of user visits, which makes computing the offline solution
of the problem fast and practical.  The resulting allocation plan is sent
to all the ad server machines.

\begin{figure}[t]
\centering
\includegraphics[width=3in]{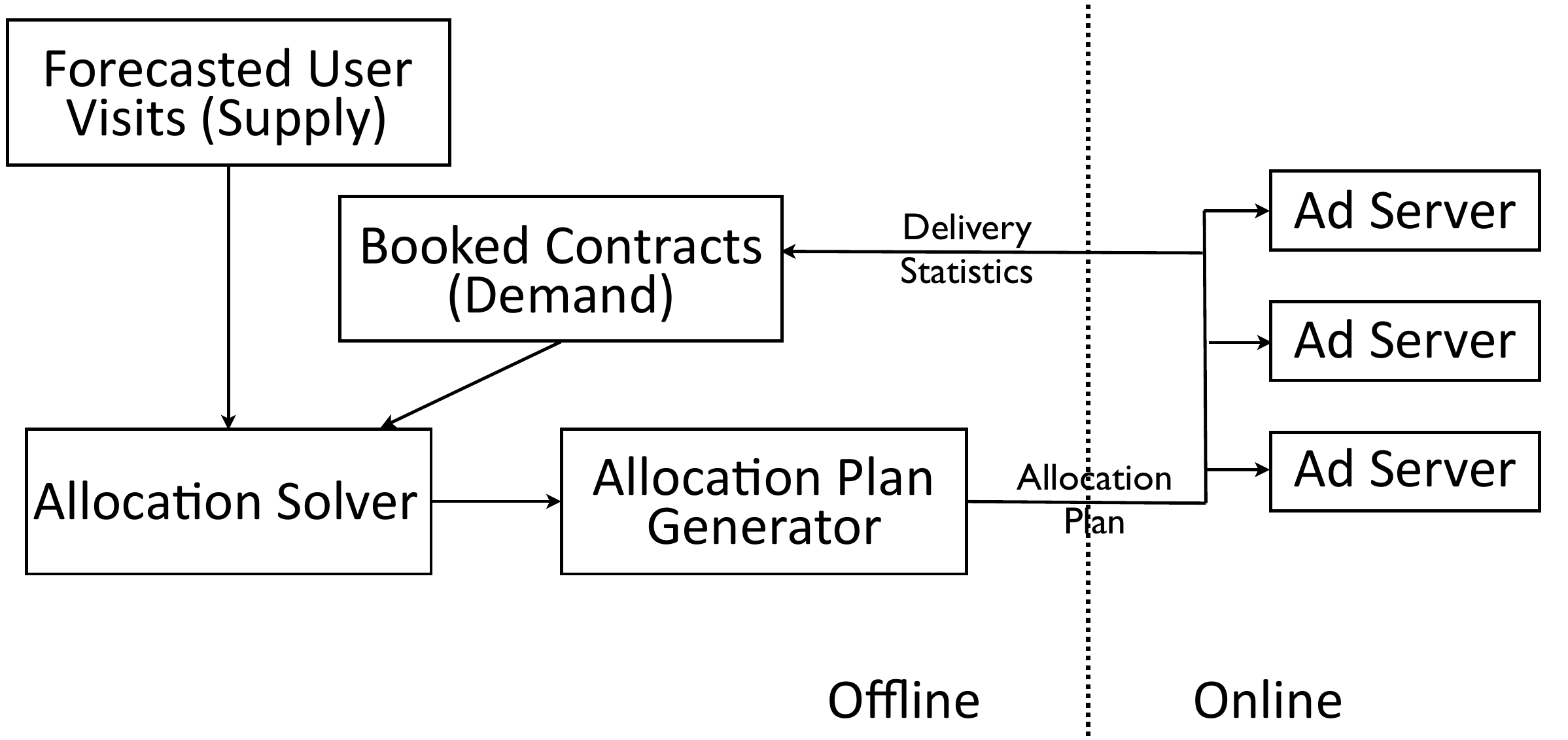}
\caption{Proposed System Architecture}
\label{fig:arch}
\end{figure}

Each ad server machine loads up the allocation plan at start-up time,
and stores the plan in memory. Then, whenever it needs to process a
user visit, it first finds the matching contracts, and then uses the allocation
plan to select a contract to show. Since the latter phase
is local (it only involves the user visit and
its matching contracts), there is no need for the ad server machines
to communicate with each other during serving, nor is there any need
for an ad server machine to maintain global state or even local
counters!

Further, as shown in the figure, there is a feedback loop
from the Ad Server to the Booked Contract Demand. Specifically,
the demand of a contract is adjusted to account for the actual
delivered impressions.
Since there can be errors in the forecast, the Allocation Plan
Generator is re-run periodically to correct for any forecast errors
based on this actual remaining demand
and a new allocation plan is then distributed to the ad servers.

In the subsequent section of the paper, we elaborate on the Allocation
Plan Generator and the Ad Server components in more detail. (A more
detailed discussion of the User Visit Forecasting system is beyond the
scope of this paper.  See~\cite{sf_paper} for a description of one such possible system.)

\section{Rate-Based Algorithms}
\label{sec:offlineAlgo}

In this section, we describe two rate-based serving algorithms.  The
first is the {\em High Water Mark Algorithm} (HWM).  The second is
an implementation of the dual-based method of~\cite{EC}, which we
will refer to as {\em DUAL}.

Both algorithms involve two phases. The offline phase, described in
Section~\ref{sec:offline}, takes as input the demand-supply forecast
graph and produces an allocation plan.  This allocation plan is just
a few numbers per contract, independent of the number of
impressions.
The online phase, described in
Section~\ref{sec:online}, repeatedly takes as input a user visit and
the set of eligible contracts for that user visit, and decides which
ad to serve based on the allocation plan.

The HWM algorithm is lightweight and fast, even in the offline
phase.  The DUAL algorithm, while lightweight and fast during the
online phase, is not nearly as fast in the offline phase.  (In our
experience, it is more than an order of magnitude slower.)  However,
it has the advantage of yielding provably optimal allocation during
serving, assuming that forecasts are perfect.  Of course, forecasts
are far from perfect in practice.  Thus, the benefits of allocation
using DUAL are somewhat mitigated.

The key to having a compact allocation plan is having a way to infer the the $x_{ij}$
values at serving time without explicitly storing them.  This can be done by creating
an $\rate_j$ fraction which represents the fraction of demand to be give to the contract,
and  which is the same for all neighboring supply nodes of a contract.  The mathematical basis
of this representation comes from the fact that an optimal solution can be represented
using these $\rate_j$ values~\cite{EC}.  Although Linear Programming solvers can be
used to find an optimal solution, due to the size of the problem and time constraints,
we describe a much faster heuristic.

\subsection{HWM Algorithm}
\label{sec:offline} Offline optimization for HWM uses a simple, but
surprisingly effective heuristic for generating
the allocation plan.
The algorithm generates a serving rate for each contract $j$, denoted $\rate_j$, together with an
{\em allocation order}.  The allocation order is set so that a contract $j$ with less eligible supply (i.e.
$\sum_{i\in\neij} s_i$) has higher priority than those with more eligible supply.  By labeling each
contract with just two numbers, the HWM algorithm creates a compact and robust allocation plan.

During online serving, every contract gets an $\rate_j$ fraction of
every impression, unless this would mean giving more than 100\% of
the impression away.  In this case, ties are broken in allocation
order: a contract $j$ coming first in the allocation order will get
its $\rate_j$ fraction, while a contract $j'$ coming later will get
whatever fraction is left, up to $\rate_{j'}$  (possibly getting
nothing for some impressions).

The algorithm itself labels every contract $j$ with its eligible supply, denoted $S_j$, which is used to
determine the allocation order; contracts with smaller $S_j$ values come earlier in the allocation order.
In order to determine the serving rate, $\rate_j$ for each $j$, the HWM algorithm does the following steps.
\begin{enumerate}\itemsep=0in
\item Initialize remaining supply $r_i = s_i$ for all $i$.
\item For each contract $j$, in allocation order, do:
    \begin{enumerate}\itemsep=0in
    \item Solve the following for $\rate_j$:
        $$\sum_{i\in\neij} \min\{ r_i , s_i \rate_j\} = d_j , $$
        setting $\rate_j = 1$ if there is no solution.
    \item Update $r_i = r_i - \min\{ r_i, s_i \rate_j\}$ for all $i\in\neij$.
    \end{enumerate}
\end{enumerate}
Notice that it is a simple manner to choose the allocation order using a different heuristic.  In our experience, the available
supply works quite well.  Contracts that have very little time left will also have very little supply and naturally become high
priority.  Likewise, highly targeted contracts also tend to have much less inventory.  This generally means they have less flexibility
in which user visits they can be served, and they are treated with higher priority.
Contrast this with a large, untargeted contract that can easily forego certain types of users visits when other highly targeted (and generally
more expensive) contracts may need those visits in order to satisfy their demands.

The compact allocation plan produced by the HWM algorithm is shown in Figure~\ref{fig:plan}.
Note that only the right hand side of the graph is sent to the Ad Servers.
In this example the allocation plan results in a feasible solution.
We note that the allocation plan is highly dependent on the supply forecast,
and an incorrect supply forecast may yield suboptimal results. We will explore this issue in detail in Section \ref{sec:robust}.

\begin{figure}[t]
\centering
\includegraphics[width=3in]{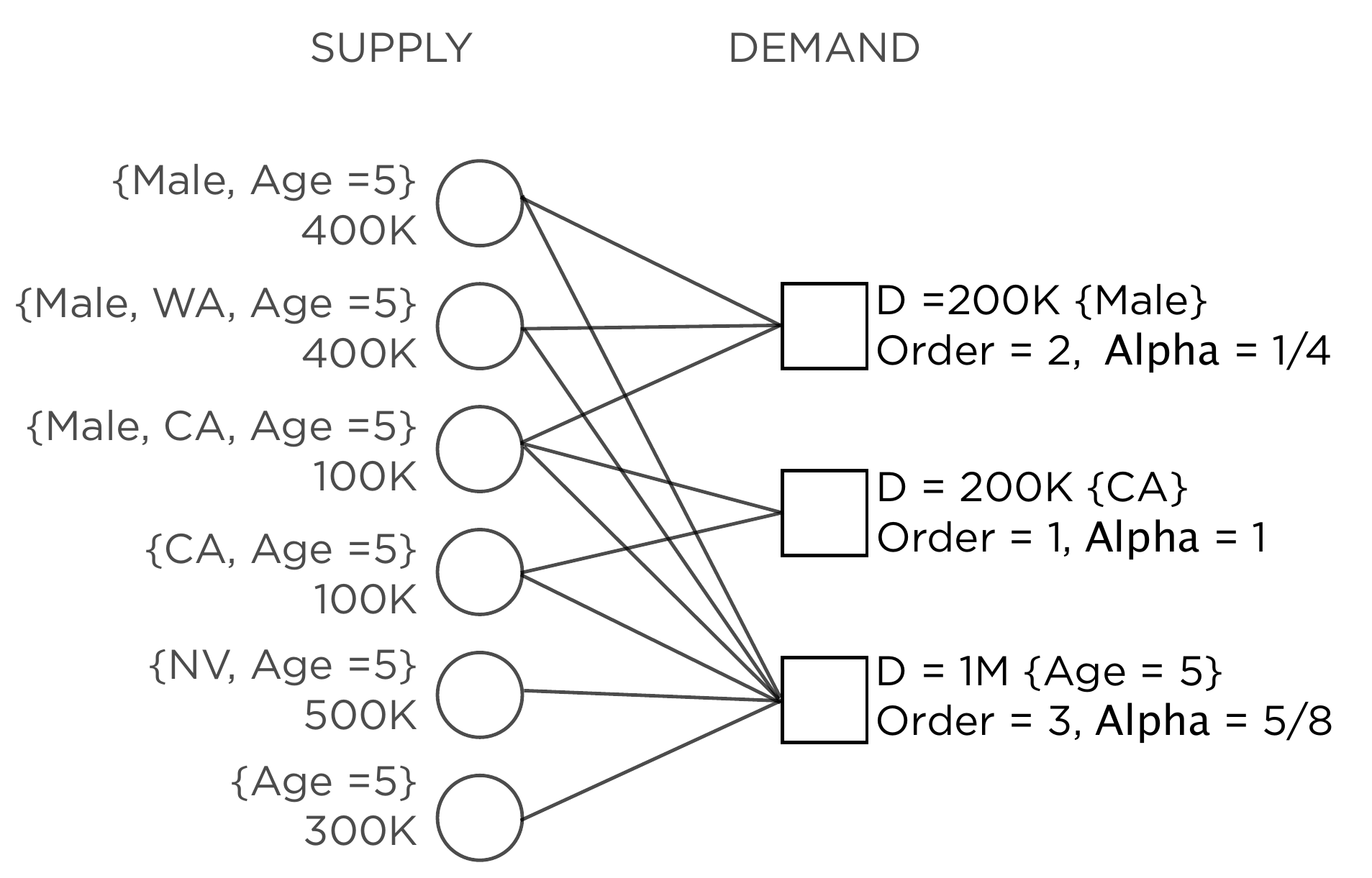}
\caption{An example compact allocation plan with the complete supply-demand graph. Note that only the data on the RHS of the figure is sent to the ad servers.}
\label{fig:plan}
\end{figure}

\subsubsection{Online Evaluation for HWM} \label{sec:online} Given an
allocation plan, which consists of the allocation order and serving
rate $\rate_j$ for each contract $j$, the role of the ad server is
to select one of the matching contracts. We present the algorithm
that the ad server follows below:
\begin{enumerate}\itemsep=0in
\item Given an impression $i$, let $J = \{c_1, c_2, \ldots, c_{|J|} \}$ be the set of matching contracts listed in the allocation order.
\item If $\sum_{j = 1}^{|J|} \rate_j > 1$, let $\ell$ be the maximum value so that $\sum_{j=1}^{\ell} \rate_{j} \leq 1$. Finally, let $\rate'_{\ell+1} = 1 - \sum_{j=1}^{\ell} \rate_j$. Note that by definition of $\ell$, $\rate'_{\ell+1} < \rate_{\ell + 1}$.
\item Select a contract $j \in [1, \ell]$ with probability $\rate_j$ and the contract $\ell+1$ with probability $\rate'_{\ell + 1}$. Note that in the case that $\sum_{j=1}^{|J|} \rate_j < 1$ with some probability no contract is selected.
\end{enumerate}

Consider the running example from Figure \ref{fig:plan}.
In the case an impression of type \texttt{\{CA, Age = 5\}} arrives, $\ell = 1$
and the middle contract with allocation order $1$ is always selected, since its serving rate is $1$.
On the other hand, suppose an impression of type \texttt{\{Male, Age = 5\}} arrives.  Then with probability $\nicefrac{1}{4}$ it is allocated to contract 1, and with probability $\nicefrac{5}{8}$ it is allocated to contract 3, with the remaining probability ($\nicefrac{1}{8}$) it is left unallocated.
(In case no contract is selected, the impression is then auctioned off in the Non-Guaranteed Marketplace.)

\subsection{The DUAL Algorithm}
The DUAL algorithm replaces the demand constraints described in
Section~\ref{sec:prob statement} with the following:
\begin{align*}\itemsep=0in
\mbox{\ \ }
        & \forall_j \ \ \ \ssum{i \in \neij} x_{ij} s_i +u_j \ge d_j &
\end{align*}
where we think of $u_j$ as the under-delivery for contract $j$.  We
additionally add the constraint that $u_j \ge 0$.  Our objective is
then
\begin{align*}\itemsep=0in
\mbox{Minimize\ \ } \ssum{j} \ssum{i \in \neij} s_i (x_{ij} -
\theta_{j})^2/\theta_{j} +
    \ssum{j} p_j u_j
\end{align*}
where $\theta_j$ and $p_j$ are fixed.  The first term attempts to
minimize the distance of the allocation from some target,
$\theta_j$.  We set $\theta_j$ to be the total demand of contract
$j$ divided by its total (forecasted) supply.  In this way, we
attempt to make the overall mix of impressions as uniform as
possible.  (See~\cite{Informs,CIKM,Ghosh,KDD} for more discussion
on this.)

The second term represents the under-delivery penalty for the
contracts.  We set $p_j$ (the penalty value per impression for
contract $j$) to be 10 for each contract.  This heuristically chosen
value performs well in practice.

In order to produce an allocation plan, DUAL simply reports the dual
values of the modified demand constraints for an optimal solution
--- in practice, these are actually the duals for an approximately
optimal solution.  Note that there is one dual value per contract.

\subsubsection{Online Evaluation for DUAL}
For each impression, we find the set of eligible contract.  This,
together with the dual values for their respective demand
constraints, is enough to compute the value of the primal solution
on the corresponding edges (i.e. the $x_{ij}$ values) computed
during the offline phase.  In fact, if the set of impressions seen
during serving is the same as the set of impressions used during the
offline phase, this method will faithfully reproduce exactly the
same allocation. Mathematically, define $g_{ij}(z) = \max\{0,
\theta_j (1+z)\}$.  For an impression $i$, let $C_i$ be the set of
eligible contracts, and let $\alpha_j$ be the dual of the modified
demand constraint for each $j\in C_i$.  Then we first solve the
following for $X$:
\begin{align*}\itemsep=0in
\ssum{j\in C_i} g_{ij}(\alpha_j - X) = 1\ .
\end{align*}
Then, we set $\beta = \max\{0, X\}$.  (It turns out that $\beta$ is
the dual value for the supply constraint for impression $i$.)
Finally, we set $x_{ij} = g_{ij}(\alpha_j - \beta)$ for all $j\in
C_i$; note that $x_{ij} = 0$ for all $j\notin C_i$.
As shown in~\cite{EC}, this reproduces the primal solution.

As with HWM, the fractional solution is interpreted as a probability
for serving.  That is, contract $j$ is served to impression $i$ with
probability $x_{ij}$.

\subsection{Robustness to Forecast Errors}
\label{sec:robust} Since rate-based algorithms produce serving rates
rather than absolute goals, it becomes much more robust to load
balancing, server failures, and so on.  However, it becomes more
sensitive to forecast errors.  For example, if the forecast is
double the true value, the serving rate is half what it should be,
roughly speaking. In this section, we show that in an idealized
setting, frequent re-optimization mitigates this problem greatly.
Note that this simple analysis applies to both HWM and DUAL.

Let us take a simple example in which our forecast is much higher
than reality for all contracts. Initially, all of the contracts are
underdelivering, since the actual supply is lower than the
forecasted supply, the $\rate_j$ given in the allocation plan is too
low. As time passes, however, the rate begins to increase. This
increase is not because the supply forecast errors are recognized,
rather, as the expiration date for the contract approaches, the
urgency with which impressions should be served to the particular
contract increases.

\begin{table}[t]
\begin{center}
\caption{Increasing serving rate due to supply forecast errors}
\vspace{2pt}
{%
\begin{tabular}{|l|c|c|c|c|}
\hline
\multirow{2}{*}{Day} & Predicted & Remaining & Serving & Actual \\ & Supply & Demand & Rate & Delivery \\
\hline
1 & 5M & 2.5M & 0.50 & 0.4M \\
2 & 4M &    2.1M &  0.525 &     0.42M \\
3 & 3M & 1.68M &    0.56 &  0.45M \\
4 & 2M &    1.23M & 0.62 &  0.49M \\
5 & 1M & 0.74M &    0.74 &  0.59M\\
6 & &   0.15M & &  \\
\hline
\end{tabular}}
\label{tab:reopt}
\end{center}
\end{table}

To illustrate this point, consider a five day contract $j$ with forecasted supply of $1M$ impressions per day, and
total demand of 2.5M impressions. If this is the only contract in the system, we should serve it at the rate of
$\rate_j = \nicefrac{2.5M}{5M} = 0.5$. If the supply forecasts are perfect, then the contract will finish delivering exactly at the end of the last day. On the other hand, suppose that the actual supply is only 800K impressions per day (a 20\% supply forecasting error rate).  Then at the end of the first day the new serving rate would be recomputed to = $\rate_j = 2.1M/4M = 0.525$.  Continuing with the example, we observe the behavior in Table \ref{tab:reopt}, which shows that after 5 days the contract would have underdelivered by $\nicefrac{.15M}{2.5M} = 6\%$.  Note that without this re-optimization, the underdelivery would be exactly 20\%, thus frequent reoptimizations help mitigate supply forecast errors.

We formalize the example above in the following theorem.
\begin{theorem}
Suppose we have $k$ optimization cycles, the supply forecast error rate (defined
to be 1 minus the ratio of real supply to predicted supply) is $r$, and the contract is never infeasible.
In the case that $r>0$, the underdelivery is positive and bounded above by $\frac{r+r^2}{k^{1-r}}$.
In the case that $r<0$, the overdelivery is positive and bounded above by $\frac{|r|}{k^{1-r}}$.
\end{theorem}
\begin{proof}
In the $i$-th round, the optimizer will allocate a $\nicefrac{1}{(k-i+1)}$ fraction of
supply, of which a $(1-r)$ fraction will be delivered. Therefore, the desired demand
will decrease by a factor of $1 - \frac{1-r}{k-i+1}$. Therefore, the supply available
 in the last round (i.e. the underdelivery) is:
\begin{align*}
&\prod_{i=1}^k \left(1 - \frac{1-r}{k-i+1}\right)
     = \prod_{i=1}^k \left(\frac{k-i+r}{k-i+1}\right)
     = \frac{r}{k} \prod_{i=1}^{k-1} \left(1+\frac{r}{i}\right)
\end{align*}
Notice that when $r>0$, this value is positive (i.e. meaning we underdeliver), while
when $r<0$, this value is negative (i.e. meaning we overdeliver).
First, consider the case that $r>0$.  We use the fact that
$\sum_{i=2}^{k-1} \frac{1}{i} \le \ln (k-1) < \ln k$.
We have
\begin{align*}
     & \frac{r}{k} \prod_{i=1}^{k-1} \left(1+\frac{r}{i}\right)
     \le \frac{r(1+r)}{k} \exp\left(\sum_{i=2}^{k-1} \frac{r}{i}\right) \\
     &\le \frac{r+r^2}{k} \exp\left(r\ln k \right)
      = \frac{r+r^2}{k^{1-r}}
\end{align*}
Now,
consider the case that $r<0$.  Here, we use the fact that $\sum_{i=1}^{k-1} \frac{1}{i} \ge \ln k$.
Hence, we have
\begin{align*}
     & \frac{|r|}{k} \prod_{i=1}^{k-1} \left(1+\frac{r}{i}\right)
     \le \frac{|r|}{k} \exp\left(\sum_{i=1}^{k-1} \frac{r}{i}\right) \\
     &\le \frac{|r|}{k} \exp\left(r\ln k \right)
      = \frac{|r|}{k^{1-r}}
\end{align*}
\end{proof}
So even a large 2X forecast error, leading to delivery rate of 0.5 of what it should be, does not result in delivering only 50\% of
the required demand.  On a week-long contract, re-optimizing every two hours, we have $k = 84$ cycles, and our under-delivery
is only $8.2\%$.  Forecast errors in the other direction (which would lead to over-delivery) are even less severe.
A 0.5X forecast error, leading to delivery rate of twice what it should be, does not result on 100\% over-delivery.
Instead, the over-delivery is much less than 1\%.

\subsection{Feedback-Based Correction}
\label{sec:feedback}
In the previous section we showed how the supply forecast errors are mitigated by frequent reoptimization.
Although this guards against large underdeliveries, the offline optimization is slow to recognize the forecast error,
and the delivery rate increases very slowly in the beginning, even when faced with very large errors. An alternative solution is to
employ a feedback system that would recognize the errors and correct the allocations accordingly.

Such a feedback system falls squarely in the realm of control theory, and one can imagine many feedback controllers that would achieve this task.
We present a very simple protocol, which we will show performs quite well in practice.  The protocol is parameterized by two variables: $\delta$, which
controls the allowed slack in delivery, and a pair $(\beta+, \beta-)$ which control the boost to the serving rate.

If, during the lifetime of a contract,  the contract is delivering
more than $\delta$ hours behind, the feedback system increases the
reported demand left for the contract by a factor of $\beta-$. For
example, let $\delta = 12$, and consider a 7 day contract for 70M
impressions. Ideally, the contract should deliver 3M impressions by
the end of day 3. If, say, only 2M impressions were still not
delivered by noon on day 4, the feedback system would increase the
demand left (5M = 7M-2M) by a factor of $\beta+$, making it
5M$\times\beta+$.
 If on the other hand, the contract had already delivered 3M impressions by noon on day 3, the feedback system would decrease
 its demand left by a factor of $\beta-$. We will show experimentally in Section \ref{sec:exp} that this, admittedly simple,
 system performs very well in real world scenarios.

\section{Experimental Setup}
\label{sec:exp}
To validate our approach we implemented our algorithms and evaluated
them using a large snapshot of real guaranteed delivery contracts
and a log of two weeks of user visits.
We describe the setup in detail below.

\subsection{Data flow}
Figure \ref{fig:exp-setup} shows the experimental setup flow. To
simulate the offline procedures, we use a snapshot of guaranteed
contracts stored in a contract database. A supply forecasting system
using historical information produces predicted user visits matching each of
the contracts in the database. The forecasted user visits along with
the contracts are passed to a graph generator which generates a
corresponding snapshot of the supply-demand bipartite graph. Finally, the
compact allocation plan is generated as described in Section
\ref{sec:offline}.

To simulate the online environment, an ad server receives the
compact allocation produced by the offline phase. We only used one
ad server for each experiment since all the ad servers are
completely symmetric with respect to the allocation plan they
receive and the decisions they make (this symmetry is, in fact, one
of the key advantages of the proposed architecture with respect to
load balancing, etc.). A set of user visits recorded from real ad
serving logs is then streamed through the ad server, which
 uses the allocation plan
based reconstruction procedure detailed in
Section~\ref{sec:offlineAlgo} to select the appropriate contract.
Finally, we record the number of user visits allocated to each
contract.

\begin{figure}[t]
\centering
\includegraphics[width=0.4\textwidth]{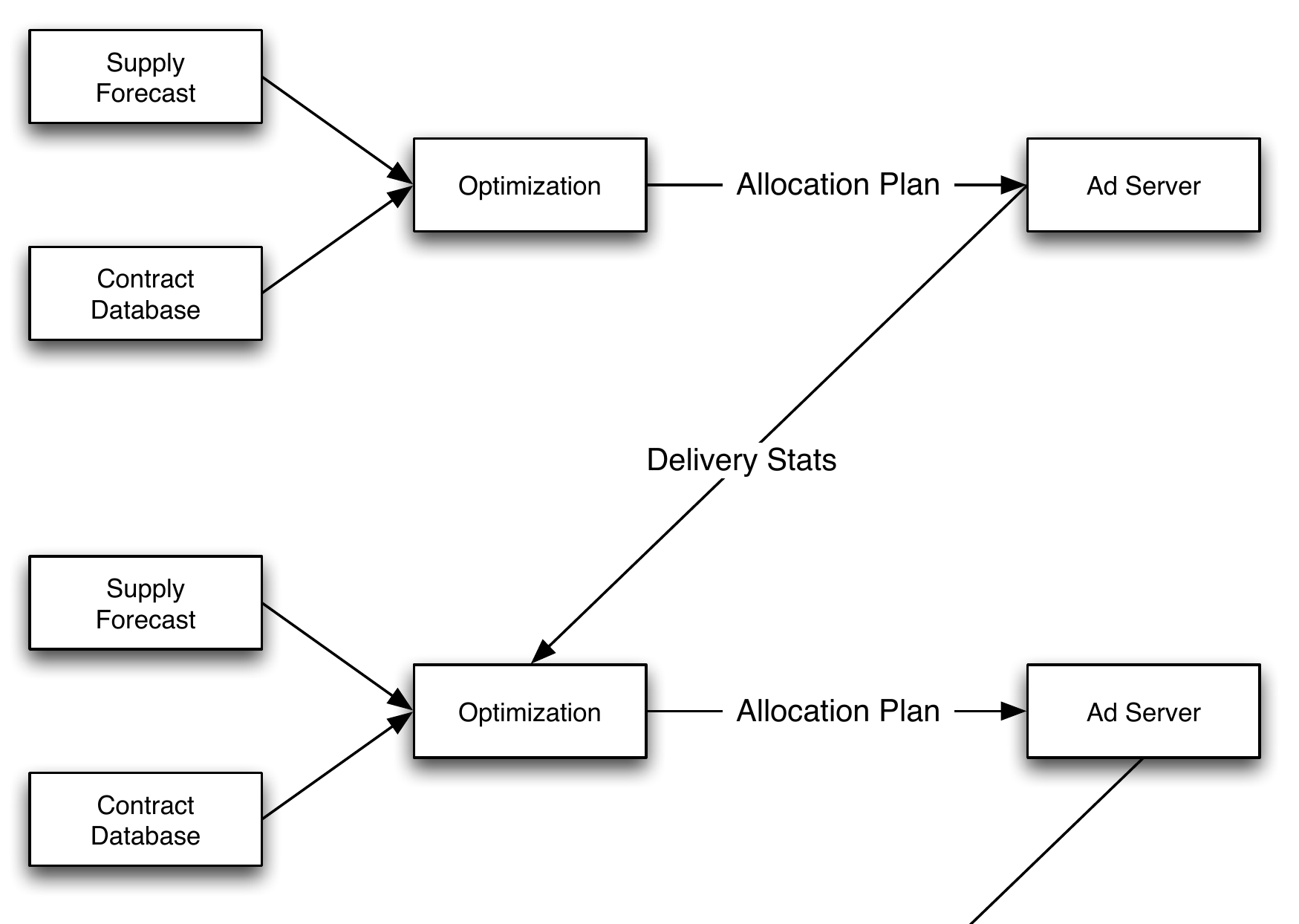}
\caption{Experimental Setup}
\label{fig:exp-setup}
\end{figure}

\subsection{Metrics} The main metric for ad serving is total
delivery rate, defined as the total number of impression delivered,
divided by the total number of impressions booked (of course for a
single contract, the number of delivered impressions is bounded by
its demand). Since the contracts are sold in a guaranteed fashion,
underdelivery represents not only unrealized revenue (since one
cannot charge for impressions not delivered), but also bad will
towards the advertisers, as they typically rely on guaranteed
delivery contracts for major marketing campaigns (e.g. new movies,
car model launches, etc.). We will measure our results as
percentage improvement in underdelivery over the baseline system.
For example, if the baseline system had underdelivered on 10\% of
impressions, and the HWM approach underdelivered on 7\%, this would
result in a 30\% improvement.

A secondary measure of the quality of the delivery is the {\em
smoothness} of the delivery. Intuitively, a week long contract
should delivery about \nicefrac{1}{7} of its impressions on every
day of the week; it should {\em not} deliver all of the impressions
in the first day (or the last day) of the campaign. Of course, due
to the interaction between different advertisers, such perfect
smoothness is not always possible, however the ad server should
strive to deliver impressions as smoothly as possible.

To measure smoothness formally, denote by $y_{j}(t)$ the total delivery to contract $j$ at time $t$, and $y^*_j(t)$ the optimal smooth delivery at time $t$.
For example, for a week long contract for 7M impressions that begins at noon on Sunday, $y^*_j(t)$ would be 1M at noon on Monday, and 5.5M at 11:59PM on Friday night.
We can then denote the smoothness of a contract at time $t$, $\sigma_j(t) = 100 \cdot \frac{y_j(t) - y^*_j(t)}{d_j(t)}$. Finally, to reduce smoothness to a single
number, let $\sigma^f(t)$ denote the $f$-th quantile of $\sigma_j(t)$ over all contracts $j$. We will typically be interested in smoothness at the 75th and
95th percentiles. Finally, we will report $\sigma^{f} = \max_{t} \sigma^f(t)$.
So, for example, a smoothness of score of 4.5 at 95th percentile means that 95\% of contracts were never over-delivering by more than 4.5\% compared to the optimal
smooth delivery.

\section{Comparison of HWM and Baseline}
\label{sec:hwm and base}
 Our first set of experiments compare a
production-quality baseline to the HWM algorithm.

\subsection{Data} We used a two week period as
our simulated testing time period. From the many real contract
active during that time window, we selected
a representative set of 1986, 
which included a mix of day long, multi day and multi week long contracts; ranging from very targeted to broad,
untargeted contracts.

On the impression side, we use actual ad serving logs to produce a set of real impressions for our test.
For scalability, we downsample at a rate greater than 1\%, resulting in 20M impressions total.
(The real traffic rate is impossible to simulate without utilizing hundreds of machines.)
Even at such a high sampling rate, the full two week simulation requires about tens hours of processing to evaluate.

\subsection{Baseline} 
We used the current guaranteed delivery system, named {\em Base}, as a baseline for our experiments.
This is a system that
aims to maintain a smooth temporal delivery on all of the contracts. At a high level, it frequently
checks and increases the serving rate of contracts that are falling behind their ideal delivery numbers,
and decreases the rate of those that have gotten too far ahead. The system is very robust to traffic changes and
supply forecast errors, but, in contrast to the proposed HWM approach, it is reactive, rather than proactive.
For example, if a certain slice of inventory is sold out in the coming days, the offline optimization will recognize the fact
and may frontload a competing contract targeting that area. The reactive system, on the other hand, will not recognize the error
until it is too late.

\subsection{Experimental Results: Regular Serving}
The experimental results are shown in Table \ref{tab:results-reg}.
As we can see, the overall delivery is much better with the HWM
approach, as the underdelivery rate falls by 53\%. However, the
extra delivery comes at the cost of smoothness
--- of the 25\% of contracts that frontloaded, most did so by almost 3
times the frontloading rate of Base. Recall that Base specifically
targets perfectly smooth delivery.  So although the smoothness of
HWM here is higher than ideal, it is still not dramatically
unacceptable.
\begin{table}[h]
\begin{center}
\caption{Total delivery for Base and HWM approaches. Note that HWM reduces underdelivery by 53\%, but has worse overall smoothness. }{%
\begin{tabular}{|l|c|c|c|}

\hline
Algorithm &  Delivery Improvement & $\sigma^{75}$ & $\sigma^{95}$ \\
\hline
\hline
Base & --- & --- & --- \\
HWM & 53\% & +288\% & +634\% \\
\hline
\end{tabular}}
\label{tab:results-reg}
\end{center}
\end{table}

\subsection{Experimental Results: Supply Forecast Errors}
\label{sec:sf-error}
The main culprit behind the inferior smoothness are the supply forecast errors. Recall, that the allocation plan,
which determines the serving rates, $\rate_j$ for each contract $j$, is computed offline using a forecast of
the expected supply. In the case that the forecast is lower than the actual supply, the system will use a higher rate of
serving than necessary, resulting in frontloading, that is delivering impressions ahead of the smooth linear goal. On the other
hand, if the forecast is higher than actual supply, this will result in backloading, and potential underdelivery of the contracts.

The description of the supply forecasting system is beyond the scope of this work, but the distribution of errors is shown in
Figure \ref{fig:sf-error}. (The scale of the y-axis is suppressed for confidentiality reasons.)
Although many of the forecasts are accurate, we can see that a significant fraction have predictions that are inaccurate enough
to potentially cause serving errors.

\begin{figure}[h]
\centering
\includegraphics[width=0.45\textwidth]{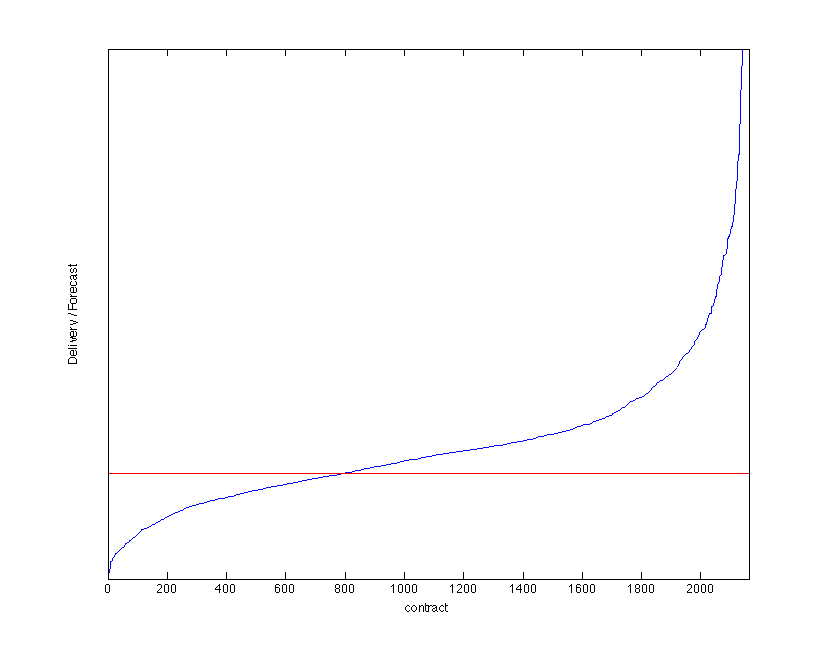}
\caption{A distribution of supply forecast errors..}
\label{fig:sf-error}
\end{figure}

\subsubsection{Smoothness}

To improve the smoothness of delivery we consider the feedback
system described in Section \ref{sec:feedback}. We set the $\delta$
parameter to 4 hours, and $\beta-$ to $10$, implying that a contract
that frontloads by more than 4 hours will get its delivery rate
reduced to 10\% of its original value. This effectively reduces  the
frontloading performed by the system. We denote this algorithm HWM+,
and show the results in Table \ref{tab:results-tw}.

\begin{table}[h]
\begin{center}
\caption{Total delivery for Base and HWM approaches. Note that HWM+ obtains comparable smoothness to the baseline,
while still drastically improving overall delivery. }{%
\begin{tabular}{|l|c|c|c|}

\hline
Algorithm & Delivery Improvement & $\sigma^{75}$ & $\sigma^{95}$ \\
\hline
\hline
Base & --- & --- & --- \\
HWM & 53\% & +288\% & +634\% \\
HWM+ & 40\% & -66.0\% & +20.9\% \\
\hline
\end{tabular}}
\label{tab:results-tw}
\end{center}
\end{table}

Observe that HWM+ slightly decreased total delivey compared to HWM. 
However, it still dramatically improves delivery compared to the current baseline system, and
its smoothness is on par --- for the majority of the contracts the delivery is even smoother.

\subsubsection{Stress Test}
We have already observed that an accurate supply forecast is an integral component of the system.
As we saw in Section~\ref{sec:robust}, the HWM algorithm, when allowed to periodically re-optimize based on the latest delivery
statistics, gracefully handles forecast errors.
In this subsection, we empirically verify that the system indeed performs well even with large forecasting errors.

We altered the supply forecast to be
biased towards underdelivery. To do this, we
effectively doubled the supply forecast for all of the contracts.  In this case, the supply forecasting system was {\em over}-predicting
the available supply on more than three-quarters of the contracts. 
Although such a high supply forecasting error rate should be unlikely to occur in practice, it is a good demonstration of the
system's robustness, even in adverse settings.

The results for the supply forecasting stress test are shown in Table \ref{tab:sfstress}. Note that since the current system, Base,
does not use supply forecasting to make its serving decisions, its delivery is unaffected. It is not surprising that the overall
delivery of HWM+ is lower; on the other hand, we observe that frequent reoptimization (the re-optimization happens 12 times per day)
limits the impact of such high supply forecast errors. Moreover, even in such a stressed state, the HWM+ algorithm performs
better than Base.

Finally, just as we used a simple feedback system to remedy frontloading errors,
we use the same system system to battle the underdelivery. Keeping $\delta$ at 4 hours, and $\beta-$ at 10, we set $\beta+$ to 1.5. In this scenario, when the
delivery is 4 hours behind (with respect to the linear goal) we artificially increase the demand of the contract by 50\%.
This increases the urgency on the contract, which leads it to deliver at a higher rate in the short term. Note that once the
total delivery is within two cycles of the linear goal, the demand is no longer artificially boosted. We call this algorithm
HWM++ and show its results in Table \ref{tab:sfstress}.  This admittedly simple feedback mechanism almost completely counteracts
the effects of the supply forecast errors, resulting in near optimal delivery.

\begin{table}[h]
\begin{center}
\caption{Total delivery for the Base and HWM approaches. Note that HWM++ obtains comparable smoothness to the baseline,
while still drastically improving overall delivery. }{%
\begin{tabular}{|l|c|c|c|}

\hline
Algorithm &  Delivery & $\sigma^{75}$ & $\sigma^{95}$ \\
& Imprv. & & \\
\hline
\hline
Base & --- & --- & --- \\
HWM & 53\% & +288\% & +634\% \\
HWM+ & 40\% & -66.0\% & +20.9\% \\
HWM
(2x SF errors) & 6\% & -5.0\% & +92.8\% \\
HWM++ (2x SF errors) & 47\% & +2.6\%\% & +116\% \\
\hline
\end{tabular}}
\label{tab:sfstress}
\end{center}
\end{table}

\section{Comparison of DUAL and HWM}
\label{sec:hwm and dual}
 We ran further experiments on DUAL
and HWM, showing that the two perform comparably.

\subsection{Data}
We used a snapshot of actual data from two different time periods on
two different areas of traffic to run our simulations.  The first
uses a week-long period of actual ad logs, downsampled at a rate of
10\%. We further used a set of real contracts active during that
time period.  The second also uses a week-long period of actual ad
logs, but downsampled at a rate of 0.25\%.  (The traffic on this
area was much greater.)

\subsection{Experimental Results}
This experiment compares the HWM and DUAL algorithms, both with and
without feedback.  Our results report the delivery and smoothness
for contracts that finished during the active time period, as well
as the smoothness of the unfinished contracts during the active
period. They are shown in Table~\ref{tab:hwmvdual}, with ++ denoting
algorithms with feedback. As we see, the general trend for
smoothness is the same for both finished and unfinished contracts.
\begin{table}[h]
\begin{center}
\caption{Comparison of HWM and DUAL on the first data set}{%
\begin{tabular}{|l|c|c|c|c|c|}

\hline
Algorithm &  Delivery & $\sigma^{75}$ & $\sigma^{95}$ & $\sigma^{75}$ \\
& Imprv. & (finished) & (finished) & (unfinished) \\
\hline \hline
HWM++ & --- & --- & --- & --- \\
DUAL++ & +4.82\% & +16.96\% & +2.44\% & +3.36\% \\
HWM & -43.89\% & +143.79\% & +367.30\% & +107.62\%\\
DUAL & -3.85\% & +173.16\% & +375.44\% & +126.89\%  \\
\hline
\end{tabular}}
\label{tab:hwmvdual}
\end{center}
\end{table}

We also ran the comparison between HWM and DUAL on our larger data
set.  Here, we compare only the feedback algorithms.  The results
are shown in Table~\ref{tab:hwmvdual2}.

\begin{table}[h]
\begin{center}
\caption{Comparison of HWM and DUAL on the second data set}{%
\begin{tabular}{|l|c|c|c|c|c|}

\hline
Algorithm &  Delivery & $\sigma^{75}$ & $\sigma^{95}$ & $\sigma^{75}$ \\
& Imprv. & (finished) & (finished) & (unfinished) \\
\hline \hline
HWM++ & --- & --- & --- & --- \\
DUAL++ & -11.79\% & +95.24\% & +155.99\% & +95.24\%  \\
\hline
\end{tabular}}
\label{tab:hwmvdual2}
\end{center}
\end{table}

Clearly, the algorithms perform much better, both in terms of
under-delivery and smoothness, when using feedback.  However, we see
two somewhat surprising things.  First, the under-delivery rate for
both HWM and DUAL are comparable.  In fact, HWM does somewhat better
in the second data set.  Note that both algorithms have very low
under-delivery rates (although we cannot display the absolute values
for confidentiality reasons), hence a difference of 5 or 10\%
translates to a very small absolute difference.

Second, DUAL is actually worse in terms of smoothness, despite the
first term of its objective. Indeed, we see the first term
of the DUAL objective (which attempts smoothness over all
impressions) evaluates to a better value for the online DUAL
allocation than the online HWM allocation. Yet the temporal
smoothness that we report is worse.

Our belief is that this has three contributing factors: (1) The
first term of the DUAL objective attempts smooth allocation across
all impressions, with time being just one dimension among many (with
others, like gender, age, interests, and so on being just as
important), (2) HWM's first-come-first-served approach tends to give
many contracts very smooth allocation, with the contracts allocated
near the end having potentially blocky distributions, while DUAL
tends to distribute non-smoothness among all contracts, and (3)
supply forecasting errors make it difficult for an optimal offline
solution to translate to a truly smooth online allocation.

In general, many of the advantages of generating an optimal serving
plan using a method like DUAL are mitigating because of forecasting
errors and the feedback needed to correct for those errors. Thus,
HWM provides a very practical solution in terms of speed and
performance.

\section{Related Work}
\label{sec:related}
There are three major related problems in the
context of guaranteed display advertising: the allocation problem, the
booking problem, and the ad serving problem (the focus of this
paper). In the allocation problem, the goal is to find an allocation
of supply (user visits) and demand (contracts) so as to meet an
optimization objective. In this context, 
\cite{Ghosh} and
~\cite{Informs} propose various objectives to smooth the
allocation of supply to demand so as to avoid the ``lumpy'' solutions
characteristic of simple linear programming
formulations~\cite{Nakamura}. In the booking problem, the goal is to
efficiently determine whether a new contract can be booked without
violating the guarantees of existing booked contracts.
The work of ~\cite{Feige}, \cite{Aleai}, \cite{Negruseri} and~\cite{Radovanovic}
consider various techniques and relaxations to solve this problem
efficiently.

The guaranteed ad serving problem can be viewed as an online
assignment problem, where some information is known about the future
input, but the input is revealed in an online fashion one vertex at a
time, and the goal is to optimize for an allocation objective. As such,
this is related to the literature on online matching. In a celebrated
result,~\cite{Karp} showed that a simple randomized
algorithm that finds a matching of size $n(1 - \frac{1}{e})$. However,
Karp et al.'s solution only works for simple matching, requires
maintaining counters across distributed servers, and does not apply to
more general allocation objectives that are considered in display
advertising. Feldman et al.~\cite{Feldman} extend Karp et al.'s
results and show that if some information about the future is
available, the $n(1 - \frac{1}{e})$ bounds can be further
improved. The proposed solution works by solving two problems offline,
and using both solutions online. However, the entire solution graph,
including both supply and demand nodes and the edges, has to be sent
to the online server. Consequently, the plan as described is not compact, cannot be
based on sampled user visits, and cannot generalize to unseen user
visits - all of which are crucial requirements for guaranteed ad
serving.

Another area of work that is related to the online allocation problem
is two-stage stochastic optimization with recourse. In this model, the
future arrival distribution is assumed to be known in advance, and the
goal is to find an allocation that best matches the distribution.
However, all of the algorithms (e.g.,~\cite{Eli}) assume that the
input is revealed in two stages, at first just the distribution of the
future input, and then the full input is revealed. Consequently, these
algorithms do not work in an online fashion, and the solution computed
in the first stage does not generalize to unseen input.

The ad serving solution presented in this paper builds upon a
mathematical result by Vee et al.~\cite{EC} that shows that for a
fairly general set of allocation objectives, the {\em optimal}
solution (except for sampling errors) can be compactly encoded as
just the dual values of the demand constraints.  Note that the
approach of Devanur and Hayes~\cite{DH09}, which also uses
dual-based methods, does not apply in our setting since it cannot
handle ``at least'' constraints.

We leverage the result of~\cite{EC} and propose a novel and
efficient High Water Mark (HWM) algorithm that can be used to
rapidly solve the allocation problem for large graphs and produce a
compact, generalizable, rate-based allocation plan for a distributed
ad server. Further, we introduce a rapid feedback loop, and show how
this can form a theoretical basis for dealing with supply
forecasting errors using the HWM algorithm. Finally, we also present
an experimental evaluation of the proposed approach using real-world
data.

\section{Conclusion \& Future Work}
\label{sec:conclusion} We have presented a new ad serving system
that relies on a compact allocation plan to decide which matching
contract to show for every user visit. This plan-based approach is
stateless, generalizable, allows for very high throughput and at the
same time performs better than the current quota-based system. While
the approach is sensitive to supply forecast errors, we have shown
that frequent re-optimization helps correct supply forecast errors,
and diminishes the impact that these errors have on delivery.

These rate-based algorithms show clear improvement in under-delivery
for ad serving, while maintaining comparable smoothness.  However,
there is no clear winner between HWM and DUAL. %
Although DUAL performs better under perfect forecasts, its
performance in practice is severely impacted by forecasting errors
and feedback corrections. Thus, HWM is quite attractive due to its lightweight
implementation and speed, coupled with good overall performance.

\bibliographystyle{plain}
\bibliography{as}
\end{document}